\PassOptionsToPackage{dvipsnames}{xcolor} 
\documentclass{fmcad}

\PassOptionsToPackage{dvipsnames}{xcolor} 
\usepackage[dvipsnames]{xcolor}
\usepackage{amsmath, amssymb}   
\usepackage{graphics}
\usepackage{wrapfig}
\usepackage{caption}
\usepackage{verbatim}
\usepackage{paralist}
\usepackage{cancel}
\usepackage{xspace}
\usepackage{cleveref}
\usepackage{mathtools}
\usepackage{algorithm}
\usepackage{algorithmicx,algpseudocode}
\usepackage{subcaption}
\usepackage{cite}
\usepackage{makecell}
\usepackage{booktabs}
\usepackage{multirow}
\usepackage{xcolor}
\usepackage{listings}
\usepackage{bm}
\usepackage{cuted}

\definecolor{navy}{RGB}{0,0,128}
\definecolor{dodgerblue}{RGB}{30,144,255}

\newcommand{\relu}{\text{ReLU}\xspace{}}
\newcommand{\mysubsection}[1]{\medskip\noindent\textbf{#1}}
\newcommand{\sat}{\texttt{SAT}}
\newcommand{\unsat}{\texttt{UNSAT}}
\newcommand{\rn}[1]{\mathbb{R}^{#1}}
\newcommand{\nn}{\mathcal{N}}
\newcommand{\theory}{$\mathcal{T}$}
\newcommand{\conf}{$\mathcal{C}$}

\newcommand{\assert}{\textcolor{red}{assert}}
\newcommand{\userule}{\textcolor{blue}{rule}}
\newcommand{\step}{\textcolor{red}{step}}
\newcommand{\cl}{\textcolor{OliveGreen}{cl}}
\newcommand{\premises}{\textcolor{blue}{premises}}
\newcommand{\args}{\textcolor{blue}{args}}

\newtheorem{theorem}{Theorem}
\newtheorem{example}{Example}

\usepackage{tikz,ifthen,pgfplots}
\usetikzlibrary{arrows,trees,backgrounds,automata,shapes,decorations,plotmarks,fit,calc,positioning,shadows,chains}
\tikzstyle{every pin edge}=[<-,shorten <=1pt]
\tikzstyle{neuron}=[circle,fill=black!25,minimum size=17pt,inner sep=0pt]
\tikzstyle{input neuron}=[neuron, fill=green!50]
\tikzstyle{output neuron}=[neuron, fill=red!50]
\tikzstyle{hidden neuron}=[neuron, fill=blue!50]
\tikzstyle{merged neuron}=[neuron, fill=orange!50]
\tikzstyle{annot} = [text width=6em, text centered]
\tikzstyle{nnedge} = [-{stealth},shorten >=0.1cm, shorten <=0.05cm,line width=0.8pt,black]
\tikzstyle{proofRoot} = [rounded rectangle, fill=orange!30]
\tikzstyle{leafNode} = [rounded rectangle, fill=red!30]
\tikzstyle{proofNode} = [rounded corners, fill=orange!10, draw=black!85, opacity=0.8, text width = 4.2cm, text height= 1.2cm]
\tikzstyle{lemmaNode} = [rounded rectangle, fill=dodgerblue!30, align=center]

\tikzstyle{proofEdge} = [-{stealth},shorten >=0.1cm, shorten <=0.05cm,line width=0.8pt,black]
\usetikzlibrary{calc}

\definecolor{listingbg}{RGB}{240,240,240}
\newcommand{\writer}{AletheProofWriter}
\newcommand{\picid}{\textsc{Picid}}

\DeclareRobustCommand{\IEEEauthorrefmark}[1]{\smash{\textsuperscript{\footnotesize #1}}}

\title{\picid: Proof-Driven Clause Learning in Neural Network Verification}
\author{
	\IEEEauthorblockN{
		Omri Isac\IEEEauthorrefmark{1}\textsuperscript{\textsection},
		Idan Refaeli\IEEEauthorrefmark{1}\textsuperscript{\textsection},
		Haoze Wu\IEEEauthorrefmark{2},
		Clark Barrett\IEEEauthorrefmark{3},
		Guy Katz\IEEEauthorrefmark{1}
	}
	\IEEEauthorblockA{
		\IEEEauthorrefmark{1}The Hebrew University of Jerusalem
		\IEEEauthorrefmark{2}Amherst College
		\IEEEauthorrefmark{3}Stanford University 
	}
}

\begin{document}
	
	\maketitle
	\begingroup\renewcommand\thefootnote{\textsection}
	\footnotetext{Both authors contributed equally.}
	\endgroup
	\vspace{-0.4cm}
	\begin{abstract}
		Current Deep Neural Network (DNN) verifiers are typically designed to prioritize scalability over reliability. Reliability can be reinforced through the generation of \emph{proofs} that are checkable by trusted, external proof checkers. To date, only a handful of verifiers support proof production; and these rely on verifier-specific formats, and balance between scalability, proof detail, and the trustworthiness of their proof checker.		
		In this tool paper, we introduce \picid{}, a DNN
                verifier that produces proofs in the standard Alethe
                format for SMT solving, checkable by an independent
                checker. \picid{} implements a parallel CDCL(\theory)
                architecture that integrates the state-of-the-art,
                proof-producing CaDiCaL SAT solver with the Marabou DNN
                verifier. Furthermore, \picid{} leverages \unsat{} proofs to derive conflict clauses.
		Our evaluation shows that \picid{} generates valid proofs in the vast majority of cases and significantly outperforms existing tools that produce comparable proofs.
	\end{abstract}

	\section{Introduction}
	\label{sec:Introduction}
	Deep Neural Networks (DNNs) have achieved remarkable success across many computational tasks, yet their internal structure and parameters are largely unintelligible, which undermines trust in their decisions.
	To address this problem, the formal verification community has developed multiple algorithms for \emph{DNN verification}, which can provably determine whether a DNN complies with given specifications~\cite{BrMuBaJoLi23,AkKeLoPi19,AvBlChHeKoPr19,BaShShMeSa19,GeMiDrTsCHVe18,HuKwWaWu17,GiHenLech20,LyKoKoWoLiDa20,PuTa10,SaDuMo19,GaGePuVe19,TrBaXiJo20,WaPeWhYaJa18,ZhShGuGuLeNa20,DuNgDw25,SchFoGu22, KaLaBaBrDuFlJoKoMaNgWu25,BrMuBaJoLi23, Ba21, KhNeRoXiBeBoFiHaLeYe23, XiKrNe22}.

	DNN verifiers are typically designed with an emphasis on performance
	and scalability~\cite{KaLaBaBrDuFlJoKoMaNgWu25}, which is essential
	for handling real-world
	DNNs~\cite{ElElIsDuGaPoBoCoKa24,ReSwBuZaAmMaFrKa26,KaLaBaBrDuFlJoKoMaNgWu25}. However,
	the \emph{reliability} of these verifiers has received lesser
	attention. In practice, many DNN verifiers rely on floating-point arithmetic for internal computations, which can be numerically unstable and may compromise soundness~\cite{JiRi21,ZoBaCsIsJe21,SiSaMeSi25}. Indeed, even state-of-the-art verifiers occasionally produce incorrect results on some queries~\cite{KaLaBaBrDuFlJoKoMaNgWu25,BrBaJoWu24}. Altogether, these issues undermine the reliability of the verification process as a whole.
	
	The mainstream approach for addressing this challenge is through
	\emph{proof production}: the generation of checkable, mathematical
	artifacts that accompany verifier results. This is common practice,
	e.g., in the SAT and SMT communities~\cite{BaDeFo15,HeBi15}.  Proof
	artifacts can be checked by a reliable proof checker that either
	certifies the verification result, or exposes errors in the
	verification process. To date, only a handful of DNN verifiers support
	proof production~\cite{IsBaZhKa22,DuNg25}, and these suffer from several limitations:
	\begin{inparaenum}[(i)]
		\item
		These proofs have varying levels of
		detail: some of them are fairly high-level, omitting many details and
		placing a computational burden on the
		checker;
		\item
		The proofs are often written in solver-specific formats, making it
		difficult, e.g., to apply multiple checkers or to compare proofs
		produced by different tools;
		and
		\item 
		The  proof checkers used are often ad-hoc tools, which may
		themselves contain errors.
	\end{inparaenum}
	To the best of our knowledge, only a single proof checker for DNN verification proofs has been
	formally verified~\cite{DeIsKoStPaKa25,DeIsPaStKoKa23}; however, this
	checker lacks support for basic optimizations applied by DNN verifiers,
	thus limiting the scalability of proof-producing verifiers to
	small-scale DNNs.
	
	Here, we present a new tool that seeks to begin addressing these
	limitations. Our goal is to improve the scalability of \emph{reliable}
	DNN verifiers, namely DNN verifiers that produce proofs in a standard
	format, checkable by reliable proof checkers. To this end, we introduce \emph{\picid{}},\footnote{The code of \picid{} is publicly available at \url{https://doi.org/10.5281/zenodo.21947584}} which builds on the proof-producing DNN verifier
	Marabou~\cite{WuIsZeTaDaKoReAmJuBaHuLaWuZhKoKaBa24,IsBaZhKa22}. \picid{} invokes Marabou as a theory-solver (\theory-solver), integrating it with the CaDiCaL SAT
	solver~\cite{BiFaFaFlFrPo24} within a parallel CDCL(\theory) architecture. For
	\unsat{} queries, it reduces both Marabou proofs and CaDiCaL's
	LRAT proofs~\cite{CrLuHeHuWaKaSc17,PoFlBi23} into the standard Alethe
	format for SMT proofs~\cite{ScFlBaFo21,BaFlFoSc}, and merges them to a complete proof of unsatisfiability.  In
	particular, \picid{} introduces the following novel features:
	\begin{enumerate}
		\item A CDCL(\theory) architecture for a DNN verifier, employing the
		state-of-the-art SAT solver CaDiCaL~\cite{BiFaFaFlFrPo24}, via the
		standard IPASIR-UP interface~\cite{FaNiPrKiSzBi23,FaNiPrKiSzBi24}.
		Conflict clauses are derived based on \unsat{} proofs, using a
		variant of the proof minimization algorithm presented
		in~\cite{IsReWuBaKa26a}.
		\item Production of \unsat{} proofs in the standard Alethe format for
		SMT solvers, checkable by its standard proof checker~\cite{AnLaBa23}, and reconstructible in a verified environment~\cite{LaFlBaJaAnReScBaTi25,CoAnBaDoMe24,EkMeTiKeKaReBa17}.
		\item A \emph{Split and Conquer} (SNC) parallelization
		technique~\cite{WuOzZeJuIrGoFoKaPaBa20}, which produces proofs for
		each worker thread, and then combines them into an overall proof for
		the query at hand.
		
	\end{enumerate}
	Our evaluation indicates that \picid{} significantly outperforms an existing DNN verifier
	and an existing SMT solver that produce Alethe proofs. Furthermore,
         over $99\%$ of the proofs generated by \picid{} successfully passed their checks. For the single query for which \picid{} returned an unsound result, the proof mechanism indicated the potential steps in the verification process that could have caused the error.
	By that, \picid{} pushes forward the
	scalability of DNN verifiers, while affording high levels of
        reliability compared to existing DNN verifiers.
	
	The paper is organized as follows: In~\Cref{sec:background} we provide relevant background on DNN verification and the Alethe proof format. \Cref{sec:overview} includes an overview description of \picid{}. Then, we explain its CDCL(\theory), proofs, and parallelization modules in Sections~\ref{sec:dpllt},~\ref{sec:proofs} and~\ref{sec:snc}, respectively.~\Cref{sec:Evaluation} describes our experimental evaluation, and we conclude and discuss related and future work in~\Cref{sec:conclusion}.

	\section{Background}
	\label{sec:background}
	In this section, we provide relevant background on DNN verification and the Alethe proof format used by \picid. We assume the reader is familiar with the basics of DPLL(\theory)~\cite{GaHaNiOlTi04}.
	
	\subsection{DNN Verification}
	\label{sec:DNNverification}
	\mysubsection{Deep Neural Networks.}\cite{GoBeCu16} Formally, a DNN $ \nn:\rn{n}\rightarrow\rn{m} $ is a directed and weighted graph, structured as a sequence
	of $ k $ layers $ L_0,...,L_{k-1} $. Each layer $ L_i $ consists of
	$ s_i \in \mathbb{N} $ nodes $ v^1_i,...,v^{s_i}_i $ and $s_i$ biases
	$b^j_i\in\mathbb{Q}$. Each directed edge in the DNN connects a node from some layer to a node in the subsequent layer, i.e., is of the form
	$(v^l_{i-1},v^j_i)$. Furthermore, each edge is labeled with a weight $w_{i,j,l}\in\mathbb{Q}$.
	
	The assignment to the nodes in the first, input layer is defined by $v^j_0 = x_j$, where $\overline{x}\in\rn{n}$ is the input vector. The assignment for the
	$ j^{th} $ node in the $ 1 < i < k $ layer is computed as
	$v^j_i = f_i \left( \underset{l=1} { \overset {s_{i-1}} { \sum } }
	w_{i,j,l} \cdot v^l_{i-1} + b^j_i \right)$
	for some non-linear activation function
	$f_i:\mathbb{R}\rightarrow\mathbb{R}$.
	Neurons in the final, output layer are evaluated similarly but without applying an activation function.
	
	In this work, we focus on DNNs with the \relu{} activation function, defined as $\relu(x) \coloneq
	\max(x,0)$. Note that \relu{} has two linear phases: For non-negative inputs, the functions is the identity and regarded as \emph{active}; otherwise, the function is constant zero, and regarded as \emph{inactive}. 
	\begin{example}
		\label{ex:dnn}
		A  DNN with three layers appears in~\Cref{fig:toyDnn}. For
		simplicity, all biases are set to $0$ and are omitted.
		Given the input $[0,1,1]$ the value of $v_1,v_2$ is computed by $v_1=\relu(0-1)=0$ and $v_2=\relu(2+1)=3$, thus the network's output is $0+3=3$.
	\end{example}
	\begin{figure}[h]
		\vspace{-0.2cm}
		\begin{center}
			\scalebox{0.82}{
				\def\layersep{1.8cm}
				\begin{tikzpicture}[shorten >=1pt,->,draw=black!50, node distance=\layersep,font=\footnotesize]
					
					\node[input neuron] (I-1) at (0,-1) {$i_1$};
					\node[input neuron] (I-2) at (0,-2) {$i_2$};
					\node[input neuron] (I-3) at (0,-3) {$i_3$};
					
					\node[hidden neuron] (H-1) at (\layersep, -1) {$v_1$};
					\node[hidden neuron] (H-2) at (\layersep,-3) {$v_2$};
					
					\node[output neuron] (O-1) at (2*\layersep,-2) {$o$};
					
					\draw[nnedge] (I-1) --node[above,pos=0.4] {$1$} (H-1);
					\draw[nnedge] (I-2) --node[above,pos=0.4] {$-1$} (H-1);

					\draw[nnedge] (I-2) --node[above,pos=0.4] {$2$} (H-2);
					\draw[nnedge] (I-3) --node[above,pos=0.4] {$1$} (H-2);
					
					\draw[nnedge] (H-1) --node[above] {$-1$} (O-1);
					\draw[nnedge] (H-2) --node[above] {$1$} (O-1);

					\node[above=0.05cm of H-1] (b1) {$\relu{}$};
					\node[above=0.05cm of H-2] (b1) {$\relu{}$};

				\end{tikzpicture}
			}
		\end{center}
		\caption{A toy DNN with \relu{} activations.}
		\vspace{-0.2cm}
		\label{fig:toyDnn}
	\end{figure}

	\mysubsection{DNN Verification.} Consider a DNN $ \nn:\rn{k}\rightarrow\rn{m}$ where
	$i\in\rn{k}, o\in\rn{m}$ denote the network's
	inputs and outputs. The \emph{DNN verification problem} is the problem
	of deciding whether there exist
	$(i,o)\in\rn{k+m}$  such that
	$ \varphi(i)\wedge (\nn(i) = o) \wedge 
	\psi(o)$ for some input and output
	properties $\varphi$, $\psi$.
	If such	vectors exist, the problem is
	satisfiable and denoted \sat{}. Otherwise, it is unsatisfiable and denoted \unsat.
	As common in the literature, we consider \emph{quantifier-free linear properties}, i.e., formulas without quantifiers whose atoms are linear equations and inequalities.
	\begin{example}\label{ex:property} Consider the DNN
		from~\Cref{fig:toyDnn} and the properties
		$
		\varphi(i) \coloneq \underset{j\in{1,2,3}}{\bigwedge}0\leq i_j \leq 1, \;
		\psi(o) \coloneq o \leq -2.
		$
		This verification problem is \unsat.
		Alternatively, for the same input property $\varphi$, the same DNN, and output property $\psi'(o)\coloneq o \leq 4$, the instance is \sat, as demonstrated by the input in~\Cref{ex:dnn}.
	\end{example}
	
	A \emph{DNN verification query} is encoded as a tuple
	$Q=\langle x,A,l,u,R\rangle$, where $\langle x,A,l,u \rangle$
	constitutes a linear program~\cite{Da63,KaBaDiJuKo21}.
	Formally, let $x = \left[x_1,\ldots,x_n\right]^\intercal\in \rn{n}$ denote 
	variables, $ A \in M_{m \times n}( \mathbb{R})$ denote a constraint matrix
	and $l, u \in (\mathbb{R}\cup \lbrace \pm \infty \rbrace)^{n} $
	denote vectors of bounds. Throughout the paper, we
	use $l(x_j)$ and $u(x_j)$, to refer to the lower and upper bounds
	(respectively) of $x_j$, and denote $l \leq u$ to indicate that for each element $j$, $l(x_j) \leq u(x_j)$.
	
	$R$ represents the set of \relu{} activation constraints of the form
	$R_j\coloneq f_j=\relu(b_j)$, for variables $b_j,f_j\in V$. In addition, a fresh
	auxiliary variable $a_j$ is added per \relu{} constraint,
	representing the non-negative difference of their output and input. The
	variable is added to $x$, with an equation $a_j = f_j -
	b_j$ added to $A$, and the bounds $l(a_j)= 0$, $u(a_j) =
        u(f_j) - l(b_j)$ added to $l$ and $u$, respectively.
	The DNN verification problem is then to decide the satisfiability of
	$(A\cdot x = 0) \wedge (l\leq x \leq u) \wedge \underset{R_j\in R}{\bigwedge}f_j=\relu(b_j)$ 

	\begin{example}\label{ex:query}
		 We demonstrate this encoding with the example presented
		in~\Cref{fig:toyDnn} and the property from~\Cref{ex:property}. 
		Denote
		$i_1, i_2, i_3$ and $o$ the network's inputs and output,
		and denote $b_j,f_j$ ($j\in\{1,2\}$)  the input and
		output of neuron $v_i$, which yield the constraints $R_i$ respectively.  The affine transformations
		of the network are reduced to the constraint matrix $A\cdot x = \overline{0}$:
		\vspace{-0.4cm}
		\begin{align*}
			i_1 - i_2 -b_1 &= 0\\
			2i_2 + i_3 -b_2 &= 0\\
			f_2 - f_1 - o &= 0 \\
			f_1 - b_1 - a_1 &= 0 \\
			f_2 - b_2 - a_2 &= 0
		\end{align*}
		The property
		is encoded as the inequalities in~\Cref{ex:property}, where
		we add the
		inequalities $(0 \leq f_j)$ (for $j\in\{1,2\}$) to express the
		non-negativity of \relu{}s, and $(0 \leq a_j)$ as described. These inequalities form the bound vectors $l,u$.
	\end{example}
	
	\subsection{Proof Production} As a satisfiability problem, proofs of \sat{} are given directly by assignments that satisfy all constraints and are produced by many modern verifiers~\cite{KaLaBaBrDuFlJoKoMaNgWu25}. Proofs of \unsat{}, however, are more involved and are typically represented as a \emph{proof tree}~\cite{IsBaZhKa22}. In this tree, each node corresponds to a case split over the two phases of a \relu{} constraint, and its children represent the resulting subqueries.
	Each leaf of the proof tree contains a proof vector witnessing \unsat{} for the corresponding subspace. This representation is based on a constructive variant of Farkas’ lemma~\cite{IsBaZhKa22,Va98}, which applies to linear satisfiability problems defined by a matrix $A$ together with lower and upper bounds $l$ and $u$:
	
	\begin{theorem}[Farkas Lemma Variant (From~\cite{IsBaZhKa22})]
		\label{thm:Farkas}
		Consider the query $ A\cdot x = \bar{0} $ and $ l \leq x \leq u $, for $ A \in M_{m \times n} (\mathbb{R})$  and $ l,x,u \in \rn{n} $. The query is \unsat{} if and only if $ \exists w \in \rn{m} $ such that for 
		$w^\intercal \cdot A \cdot x \coloneq \underset{j=1}{\overset{n}{\sum}}c_j\cdot x_j$, we have that  $ \underset{c_j > 0}{\sum}c_j\cdot u_j +  \underset{c_j < 0}{\sum}c_j\cdot l_j < 0$
		whereas $w^\intercal \cdot \bar{0} = 0$. Thus,
		$w$ is a proof of \unsat{}, and can be constructed during LP solving.
	\end{theorem}
	
	Finally, each node of the tree contains a sequence of \emph{bound tightening lemmas}, which capture the derivation of variable bounds induced by piecewise-linear activation functions. Each lemma includes a \emph{causing bound} used to infer a \emph{derived bound}.
	These lemmas are proven by a proof vector, similar to the one in~\Cref{thm:Farkas}, for proving the causing bound, along with a proof rule representing the deduction of the derived bound from the causing one. The derived bound is applicable in the proof tree rooted in the node in which the lemma is learned. 
	We conclude this section with an example of such a proof.
	\begin{example}
		\label{ex:proof}
		Consider the query from~\Cref{ex:query}, and the proof tree illustrated in~\Cref{fig:proof}.
		The proof tree consists of three nodes (orange). Given the query $Q$, a single lemma (blue) is learned at the root, before splitting and deriving \unsat{} (red) on both leaves.
		The lemma is proven by the vector $\begin{bmatrix} 0 & -1 & 0 & 0 & 0 \end{bmatrix}^\intercal$ representing the equation $-2i_2 - i_3 +b_2 = 0$, equivalent to $b_2 = 2i_2 + i_3$. Using the bounds, we can deduce that $l(b_2) = 2l(i_2) + l(i_3) = 0$. Given that $b_2$ is non-negative, we deduce that the constraint $R_2$ is always active, and that for its auxiliary variable $a_2$, $u(a_2) = 0$.
		The root has two children, representing splitting over the two phases of $v_1$, both leaves with an \unsat{} vector and no lemmas. The right leaf, representing the case where $v_1$ is active, is proven with the vector  $w=\begin{bmatrix} 1 & -1 & -1 & -1 & 1 \end{bmatrix}^\intercal$. Indeed, observe that:
		\begin{align*}
		 w^\intercal \cdot A \cdot x = \; &i_1 - i_2 - b_1 \\
		  &- 2i_2 - i_3 + b_2  \\
		  &- f_2 + f_1 + o  \\
		  &- f_1 + b_1 + a_1 \\ 
		  &+ f_2  - b_2 - a_2 \\
		   &= i_1 - 3i_2 - i_3 +a_1 -a_2 + o
		\end{align*}
		The upper bound for this linear combination is then
			 $u(i_1) - 3l(i_2) - l(i_3) +u(a_1) -l(a_2) + u(o) =
			  1-0-0+0-0-2 = -1<0$.
		Therefore, we deduce that $w$ is a proof of \unsat{} as in~\Cref{thm:Farkas}.
		\unsat{} of the left leaf is proven similarly. Since all leaves are \unsat{}, we deduce the \unsat{} of $Q$.
	\end{example}

	\begin{figure}[h!]
		\centering
		\scalebox{0.8} {
			\def\ySep{1.6cm}
			\def\xSep{3cm}
			\begin{tikzpicture}[ >=stealth,shorten >=1pt,shorten <=1pt]
				\node[proofRoot] (root) at (0,0) []  {$Q$};
				\node[proofNode] (N1) at (0, -0.8*\ySep) {};
				\node[lemmaNode] (lemma1) at (0, -0.8*\ySep)  [] {$l(b_2) = 0 \rightarrow u(a_2) = 0$ \\ $\begin{bmatrix} 0 & -1 & 0 & 0 & 0 \end{bmatrix}^\intercal$};
		
				\node[proofNode] (N2) at (-\xSep, -2*\ySep) {};
				\node[leafNode, align=center] (unsatL) at (-\xSep, -2*\ySep)  [] {$\begin{bmatrix} 0 & 0 & -1 & 0 & 0 \end{bmatrix}^\intercal$};
				
				\node[proofNode] (N3) at (\xSep, -2*\ySep) {};
				\node[leafNode, align=center] (unsatR) at (\xSep, -2*\ySep)  [] {$\begin{bmatrix} 1 & -1 & -1 & -1 & 1 \end{bmatrix}^\intercal$};
				
				\draw[->, line width=0.5mm] (N1) -- (N2) node[midway, left=0.3cm] {$v_1$ inactive};
				\draw[->, line width=0.5mm] (N1) -- (N3) node[midway, right=0.3cm] {$v_1$ active};
				
			\end{tikzpicture}
		}
		\caption{A proof tree example, proving \unsat{} for the query of~\Cref{ex:query}.}
		\label{fig:proof}
		\vspace{-0.4cm}
	\end{figure}
	
	\subsection{The Alethe Proof Format}
	Alethe is a standard proof format for unsatisfiability in SMT solvers~\cite{BaFlFoSc,ScFlBaFo21}. Proofs in Alethe are expressed as sequences of \emph{proof commands}, of which two are relevant here: \begin{inparaenum}[(i)]
	\item \emph{assumptions}, which are first-order formulas importing assertions from the SMT problem; and
	\item \emph{steps}, which correspond to propositional or first-order derivations. Each step describes the application of a well-defined proof \emph{rule} that derives a clause from given \emph{premises}, i.e., previous assumptions or steps.
	\end{inparaenum}

	As common in resolution proofs, an Alethe proof concludes with a step deriving the empty clause.

	For example, our Alethe proofs extensively use the rule:
	$$\small{\bm{
	\mathtt{(\step \; ite\textunderscore i \; (\cl \; \varphi_1 \; \varphi_3) \; :\userule \; ite1 \; :\premises(ite))}}}
	$$ 
	that deduces the clause $\varphi_1, \varphi_3$ from an \texttt{ite} statement $\mathtt{ite \; \varphi_1 \; \varphi_2 \; \varphi_3}$, using the proof rule named  \texttt{ite1}. 

	Alethe proofs can be translated (i.e., \emph{reconstructed}) to proof terms in
        Isabelle/HOL~\cite{LaFlBaJaAnReScBaTi25}, and a subset of
        their rules can also be translated to Rocq~\cite{EkMeTiKeKaReBa17} and Lambdapi~\cite{CoAnBaDoMe24}. In addition, Alethe proofs can be checked by \textsc{Carcara}~\cite{AnLaBa23}, a widely used proof checker implemented in Rust. Although \textsc{Carcara} is not formally verified, it is a mature, open-source tool that checks a set of fine-grained proof rules, uses arbitrary precision arithmetic and has a significantly smaller codebase compared to \picid{}.
	For these reasons, \textsc{Carcara} provides a higher level of        
        trust than most existing DNN verifiers, including \picid{}.

	\section{Overview of \picid}
	\label{sec:overview}
	Here, we outline the main novel components of \picid{} and
	their interactions, as illustrated in Fig.~\ref{fig:overview}. In the
	subsequent sections we describe each of the components more
        thoroughly. Note that \picid{} builds upon, and extends, the existing Marabou tool;
	we use ``Marabou'' to refer to the existing baseline, and use \theory-solver to refer to Marabou when used within the context of \picid{}.
		
	\begin{figure*}[th!]
		\vspace{-0.4cm}
		\centering
		\resizebox{\textwidth}{!}{
			\begin{tikzpicture}[
				font=\sffamily,
				stage/.style={
					draw=black!85,
					fill=red!25,
					rounded corners=8pt,
					line width=1.0pt,
					minimum height=1.6cm,
					align=center
				},
				corebox/.style={
					draw=black!85,
					fill=black!5,
					rounded corners,
					line width=1.2pt,
					align=center
				},
				module/.style={
					draw=black!85,
					fill=cyan!30,
					rounded corners,
					line width=1.2pt,
					align=center
				}
				]
				
				\node[corebox, minimum width=0.55\textwidth, minimum height=0.20\textwidth] (core) {};
				
				\node[stage, minimum width=0.16\textwidth, left=0.03\textwidth of core] (pre)
				{Preprocess\\(Marabou)};
				\node[stage, minimum width=0.16\textwidth, right=0.03\textwidth of core] (post)
				{Postprocess\\(ProofWriter)};
				
				\path (core.north west) ++(0.12\textwidth,-0.06\textwidth)
				node[module, minimum width=0.18\textwidth] (cdcl)
				{CdclCore \\ Marabou \theory-solver};
				
				\path (core.north east) ++(-0.12\textwidth,-0.06\textwidth)
				node[module, minimum width=0.18\textwidth] (sat)
				{SAT Solver \\ (CaDiCaL)};
				
				\path (core.center) ++(0,-0.07\textwidth)
				node[module, minimum width=0.18\textwidth] (proof)
				{\writer};
				
				\draw[dashed, <->, line width=1.5pt] (cdcl) -- (sat)
				node[midway,above] {\textbf{IPASIR-UP}};
				\draw[dashed, ->, line width=1.5pt] (cdcl) -- (proof);
				\draw[dashed, ->, line width=1.5pt] (sat) -- (proof);
				\draw[ ->, line width=1pt] (pre) -- (core);
				\draw[->, line width=1pt] (core) -- (post);
				
				\path (core.north east) ++(0.12\textwidth,-0.03\textwidth)
				node[draw, fill=white, dashed] {$\times N$ Workers};
				
				\coordinate (in) at ($(pre.west)+(-0.07\textwidth,0)$);
				\draw[->, line width=1pt] (in) -- node[above]{Input} (pre.west);
				
				\coordinate (out) at ($(post.east)+(0.07\textwidth,0)$);
				\draw[->, line width=1pt] (post.east) -- node[above]{Output} (out);
				
			\end{tikzpicture}
		}
		\vspace{-0.4cm}
		\caption{An overview of \picid's architecture.}
		\label{fig:overview}
		\vspace{-0.6cm}
	\end{figure*}

	\sloppy
	Given an input query, \picid{} first performs a preprocessing pass, using
	Marabou modules, that transforms the given DNN and property into a set of linear and piecewise linear equations as described in~\Cref{sec:DNNverification}. It then initializes an instance of the novel \emph{CdclCore} module, which is implemented within the \theory-solver; and which 
	is responsible for orchestrating its interaction with the SAT
	solver. To
	this end, CdclCore implements the \theory-solver side of IPASIR-UP~\cite{FaNiPrKiSzBi23,FaNiPrKiSzBi24}, a general interface based on DPLL(\theory), which enables a
	modular integration of a SAT and  \theory- solvers.
	Adhering to the principles of DPLL(\theory), the SAT solver then begins to search for a satisfying Boolean assignment to its variables. During this search, the SAT solver either \emph{notifies} the \theory-solver on changes made in the Boolean level, or \emph{calls back} the \theory-solver for information.
	Upon callbacks and notifications
	from the SAT solver (e.g., \theory-propagate), CdclCore invokes the
	corresponding routines in the \theory-solver and relays the resulting information
	back to the SAT solver (e.g. assignment to Boolean
	variables).
              Proof generation is handled by the novel \emph{\writer} module, which collects proof information from both the SAT solver and the \theory-solver and emits the corresponding proof steps. 
              
              If the query at hand is \sat{}, a satisfying assignment is returned and all proof information is deleted. Otherwise, if the query is \unsat{}, the proof of unsatisfiability is readily available in the Alethe format, written in a dedicated file.
              
 			For handling complex queries, \picid{} is enhanced with a \emph{split and conquer} (SNC)~\cite{WuOzZeJuIrGoFoKaPaBa20} parallelization mechanism that splits a query into several simpler subqueries, which are independently solved with concurrent workers. This mechanism is implemented within Marabou, and extended in \picid{} to support proof-producing CDCL(\theory) solving.
              When multiple workers are used in parallel, each worker is
        initialized with its own CdclCore, \writer, \theory- and SAT
        solvers. Notably, the workers do not perform additional preprocessing, ensuring a consistent Boolean abstraction across workers. During verification, all workers write to a shared proof file. If the query is \unsat{} (i.e., all workers terminate with \unsat{}), a designated worker additionally emits proof steps that derive the overall \unsat{} result of the query from the \unsat{} results of the subqueries. Otherwise, the proof file is deleted.
	
	Currently, \picid{} can verify DNNs with the \relu{} activation function. Extending it to any piecewise linear activation function is straightforward, using an appropriate Boolean abstraction and formulation for bound propagation rules in Alethe. Furthermore, Marabou's and consequently \picid's, preprocessing is not covered by proof production, and it requires a separate analysis. Producing such proofs was known to be a challenge in SMT solving as well~\cite{BaDeFo15}. We leave both directions for future work. 
	\section{CDCL(\theory) Implementation}
	\label{sec:dpllt}
	
	\subsection{DNN Verifiers as \theory-Solvers}
	\label{sec:abstraction}
	The DNN verification problem can be formulated as an SMT problem with a single theory solver of linear real arithmetic~\cite{KaBaDiJuKo21}.
	Separating the Boolean- and theory-specific reasoning in
	this way, as is generally done in SMT solving, is beneficial for
	multiple reasons:
	\begin{inparaenum}[(i)]
		\item It provides a clean logical separation of these two distinct levels of reasoning;
		\item It enables a modular integration of modern SAT solvers into DNN
		verifiers, allowing the latter to benefit from advances in SAT
		solving; and,
		\item It paves the way for future integration with additional theory solvers.
	\end{inparaenum}
	Despite these facts, most modern DNN solvers do not follow this
	architecture. Instead, they either mix the two layers completely, or
	separate them entirely, so that very little information flows between
	the two.
	
	Here, we advocate leveraging the experience gained by
	the SMT community in recent decades, by designing a DNN verifier
	according to the DPLL(\theory) architecture.  Such a design has been attempted before~\cite{KaBaDiJuKo21,DuNgDw25,LiYaZhHu24,Eh17}; and we follow prior work and employ Boolean
	reasoning over the neurons of the DNN --- using Booleans to encode the linear phases of
	activation functions. Consequently, unlike general-purpose SMT
	solvers, only a subset of the theory constraints is represented using
	Boolean variables, while the remaining constraints are handled
	internally by the \theory-solver.
	
	Building on this foundation, we take the next step and design \picid{} as a CDCL(\theory)-based DNN verifier. Conflict-Driven Clause Learning (CDCL)~\cite{SiSk96,SiSk99,BaSc97} extends DPLL by augmenting the formula with learned conflict clauses. Each clause corresponds to the negation of a partial assignment that led to unsatisfiability, and is therefore implied by the original formula. These clauses prevent the solver from revisiting subspaces of the search space that are guaranteed to contain no satisfying assignments. In SAT solving, CDCL is known to significantly prune the search tree and substantially improve scalability~\cite{SiSk96,SiSk99,BaSc97}. SMT solvers similarly rely on CDCL(\theory), the first-order extension of CDCL~\cite{BaTi18}. Prior work in DNN verification has proposed various techniques for detecting non-trivial conflict clauses~\cite{Eh17,ZhBrHaZh24}. However, these attempts either involved repetitive invocations of an LP solver~\cite{Eh17}, or they were not implemented within a DPLL architecture~\cite{ZhBrHaZh24}.
	To the best of our knowledge, our work is the first to focus on clause learning derived directly from \unsat{} proofs. More importantly, our conflict clauses are learned alongside their corresponding  proofs.

	\subsection{Proof-Driven Clause Learning}
	\label{sec:pdcl}
	We next describe how \unsat{} proofs can be used to derive
        conflict clauses. Conceptually, whenever Marabou determines
        that a subspace is \unsat{}, the corresponding proof of
        unsatisfiability can be analyzed to extract a core set of
        input and decision constraints sufficient to imply
        \unsat{}. Since the subset of decision constraints is abstracted by the SAT solver, it can be directly translated into a conflict clause. This approach has two key advantages: \begin{inparaenum}[(i)]
		\item It improves performance speed by learning non-trivial  conflict clauses; and
		\item Each conflict clause is accompanied by a minimized proof of its correctness.
	\end{inparaenum}
	Below, we detail the adaptation of the proof-analysis algorithm of~\cite{IsReWuBaKa26a} to derive conflict clauses in \picid.

	The goal of proof analysis is reducing the number of lemmas used to derive \unsat{}, yielding more succinct proofs and, consequently, shorter conflict clauses. Conceptually, the analysis routine performs an on-the-fly backward analysis of the proof tree, starting from its leaves. When \picid{} derives \unsat{} for a subquery, the corresponding proof is analyzed to identify which bounds are used in the proof, and which originate from bound-tightening lemmas. This allows detecting lemmas that were learned eagerly but are not actually used, and can therefore be safely removed.
	A minimization step then retains a minimal subset of such
        lemmas, resulting in its \emph{\unsat-core}. The procedure is applied recursively to the proofs of
        the retained lemmas and terminates once all required bounds
        are provided directly by the input constraints or by splitting
        decisions. Termination is ensured by assigning each lemma a
        unique chronological identifier. Then, the \unsat{}-core analysis of each lemma is restricted to include only lemmas that were chronologically learned earlier.
	Adapting this process to derive conflict clauses is straightforward. Whenever a decision bound is encountered during the analysis, the negation of its corresponding (Boolean) decision variable is added to the clause. We formalize this adaptation in~\Cref{alg:lemmadeduction} and omit the full description of the underlying minimization procedure from~\cite{IsReWuBaKa26a} for brevity. To avoid redundant computation, the clause derived for each lemma is cached upon its first analysis and reused in subsequent recursive calls.
	
	\begin{example}
		\label{ex:clause}
		Recall the proof tree presented in~\Cref{ex:proof}. The right leaf is proven \unsat{} using the proof vector $w=\begin{bmatrix} 1 & -1 & -1 & -1 & 1 \end{bmatrix}^\intercal$ that generates the equation $i_1 - 3i_2 - i_3 +a_1 -a_2 + o = 0$, which requires the bounds $u(i_1), l(i_2), l(i_3), u(a_1), l(a_2), u(o)$. From these bounds, only $u(a_1)$ is learned by the decision ($v_1$ is active). This means that in this case, the clause learned is the (trivial) clause $\lnot v_1$, representing the fact the $v_1$ is inactive. Notably, the lemma learned in the root is not used for deriving \unsat{} in both leaves, so~\Cref{alg:lemmadeduction} will remove it from the proof. 
	\end{example}
	
	\begin{algorithm}[t!]
		\textbf{Input:} A DNN verification subquery $Q=\langle x,A,l,u,R\rangle$, a list of learned lemmas $Lem$, and a proof vector of \unsat{} $w$ \\
		\textbf{Output:} A conflict clause \conf. 
		\caption{\emph{PDCL():} Construct proof-driven conflict clauses}
		\label{alg:lemmadeduction}
		\begin{algorithmic}[1]
			\Statex{// Let $\underset{j=1}{\overset{n}{\sum}}c_j\cdot
				x_j$ denote the linear combination $ w^\intercal \cdot A \cdot x$}
			\State {$Deps \leftarrow \emptyset$} \Comment{An empty list of lemmas}
			\State {\conf $ \leftarrow \emptyset$} \Comment{The conflict clause}
			\For{$j\in [n] $}
			\If {$c_j > 0$ and $u(x_j)$ is learned by  $\ell^u_j\in Lem$}
			\State{$Deps \leftarrow Deps \cup \ell^u_j$}
			\ElsIf {$c_j < 0$ and $l(x_j)$ is learned by  $\ell^l_j\in Lem$}
			\State{$Deps \leftarrow Deps \cup \ell^l_j$}
			\ElsIf {$c_j > 0$ and $u(x_j)$ is learned by a decision $D$}
			\State{\conf $\leftarrow$ \conf $\lor \lnot D$}
			\ElsIf {$c_j < 0$ and $l(x_j)$ is learned by a decision $D$}
			\State{\conf $\leftarrow$ \conf $\lor \lnot D$}
			\EndIf
			\EndFor
			\State{$Deps \leftarrow proofMin(Q,w,Deps, l,u)$}
			\Comment{Minimize dependencies as in~\cite{IsReWuBaKa26a}}
			\For{$lem \in Deps$}
			\Comment{Repeat recursively}
			\State{$lem.includeInProof \leftarrow true$ }
			\State{$l',u' \leftarrow$ bounds with ID at most $lem.getID()$}
			\State{$w' \leftarrow lem.getProof()$ }
			\State{\conf' $\leftarrow PDCL(\langle x,A',l',u',R\rangle, Lem, w')$}
			\State{\conf$ \leftarrow$ \conf $\vee$ \conf'}
			\EndFor\\
			\Return{\conf}
		\end{algorithmic}
	\end{algorithm}

	\subsection{IPASIR-UP Implementation}  
	A key motivation behind our design is to allow a modular 
	integration of SAT solvers into \picid{}. To maximize modularity, we use the IPASIR-UP
	interface~\cite{FaNiPrKiSzBi23,FaNiPrKiSzBi24}, which provides a
	structured way to connect \theory-solvers with SAT solvers within a
	DPLL(\theory) framework. Though some SAT solvers support this
	interface (e.g., CaDiCaL~\cite{BiFaFaFlFrPo24}), to the best of our
	knowledge, no existing DNN verifier has implemented it thus far.
	
	The IPASIR-UP interface includes several methods, some of which serve
	as notifications from the SAT solver to the \theory-solver, informing
	it of key events during the SAT-level search process. These include
	notifications for assigned literals (\texttt{notify\textunderscore
		assignment()}), backtracking to earlier decision levels
	(\texttt{notify\textunderscore backtrack()}), and the introduction of
	new decision levels (\texttt{notify\textunderscore new\textunderscore
		decision\textunderscore level()}). Separate methods serve as
	callbacks, allowing the SAT solver to query the \theory-solver for
	relevant information. These include retrieving literals for
	propagation (\texttt{cb\textunderscore propagate()}), selecting
	decision variables (\texttt{cb\textunderscore decide()}), and
	providing explanatory clauses for previous propagations
	(\texttt{notify\textunderscore add\textunderscore
		reason\textunderscore clause\textunderscore lit()}) or
	conflicts (\texttt{notify\textunderscore add\textunderscore
		external\textunderscore clause\textunderscore lit()}).
	Additionally, CaDiCaL allows integration of its internal proof
        logging with \theory-solver proofs using two methods,
        \texttt{add\textunderscore derived\textunderscore clause} and
        \texttt{add\textunderscore original\textunderscore
          clause}. These methods inform that a Boolean resolution
        proof step is learned either by the SAT solver, or the \theory-solver, respectively.

	For \picid{} to be compatible with IPASIR-UP and
	support these methods, we maintain several data structures,
	including context-dependent ones (inspired by those used in
	Marabou~\cite{WuIsZeTaDaKoReAmJuBaHuLaWuZhKoKaBa24} and
	CVC5~\cite{BaBaBrKrLaMaMoMoNiNoOzPrReShTiZo22}).
	 These data structures include a map
	representing the Boolean abstraction
	(\texttt{satSolverVarToPlc}), a
	list of deduced literals pending propagation to the SAT solver
	(\texttt{literalsToPropagate}), a context-dependent list of currently
	assigned literals (\texttt{assignedLiterals}), and another
	context-dependent list of satisfied clauses in the current context (\texttt{satisfiedClauses}).
	These structures
	ensure that whenever the SAT solver backtracks to an earlier level in
	the search tree, the data within \picid{} reverts to the
	state it held at that level. 
	Below, we describe the implementation of each method used in the basic module. The proof methods are described in~\Cref{sec:proofs}.
	
	\mysubsection{\texttt{void notify\textunderscore assignment(const vector<int> \&lits)}}
	\label{sec:notify_assignment}
	
	This method processes a list of literals \texttt{lits} assigned by the SAT solver, notifying the \theory-solver accordingly. The \theory-solver iterates over each literal \texttt{lit} in \texttt{lits}, adding it to the internal list \texttt{assignedLiterals}. Subsequently, the corresponding piecewise-linear constraint (\texttt{plc}) is retrieved from \texttt{satSolverVarToPlc}. The phase of \texttt{plc} is fixed based on the sign of \texttt{lit}, which may lead to further tightening of \theory-variable bounds. Finally, \texttt{satisfiedClauses} is updated accordingly to the new assignments.
	
	\mysubsection{\texttt{void notify\textunderscore new\textunderscore decision\textunderscore level()}}
	\label{sec:notify_new_decision_level}
	
	This method is used to preserve the current state in context-dependent data structures.
	In particular, it stores \texttt{assigned} \texttt{Literals} and \texttt{satisfiedClauses} into their designated data structures, ensuring consistency across decision levels.

	\mysubsection{\texttt{void notify\textunderscore backtrack(size\textunderscore t new\textunderscore level)}}
	\label{sec:notify_backtrack}
	
	This method facilitates backtracking to the specified decision level \texttt{new\textunderscore level}. It restores the internal state of all context-dependent structures to align with the state corresponding to the given decision level. In particular, this includes reverting both \texttt{assignedLiterals} and \texttt{satisfiedClauses} to their respective states at \texttt{new\textunderscore level}, and clearing the pending \texttt{literalsToPropagate}.

	\mysubsection{\texttt{bool cb\textunderscore check\textunderscore found\textunderscore model(const vector<int> \&model)}}
	\label{sec:cb_check_found_model}
	
	This method verifies whether the complete boolean assignment \texttt{model}, provided by the SAT solver, satisfies all constraints encoded in the \theory-solver. The \theory-solver applies the functionality described in~\texttt{notify\textunderscore assignment} to process each previously unassigned literal in \texttt{model}. Then, the \theory-solver is invoked to determine the satisfiability of the constraints imposed by \texttt{model}. If it identifies a satisfying assignment for all variables under the phase fixings defined by \texttt{model}, the method returns \texttt{true}. Conversely, if it concludes that no such assignment exists, the method returns \texttt{false}.

	\mysubsection{\texttt{int cb\textunderscore decide()}}
	\label{sec:cb_decide}
	
	This method is invoked by the SAT solver to determine the next decision literal. In \picid{}, we integrate the VSIDS decision heuristic for SAT solvers~\cite{MoMaZhZhMa01} with the baseline heuristic used in Marabou. A score is computed for each literal based on the applied decision heuristics and the information stored in \texttt{satisfiedClauses}. Then, the literal with the highest score, representing a linear phase of an activation constraint, is returned as the decision literal.

	\mysubsection{\texttt{int cb\textunderscore propagate()}}
	\label{sec:cb_propagate}
	
	This method is invoked by the SAT solver to request additional literals for assignment that may have been deduced by the \theory-solver. The \theory-solver begins by checking whether the list of pending literals for propagation, \texttt{literalsToPropagate}, is empty. If \texttt{literalsToPropagate} is empty, the \theory-solver is executed to deduce new phase fixings, which are then added to \texttt{literalsToPropagate} for subsequent propagation.
	
	The method is repeatedly called to retrieve one literal for propagation at a time until no further literals remain. For each invocation, a single literal from \texttt{literalsToPropagate} is propagated. When \texttt{literalsToPropagate} becomes empty, the method signals this to the SAT solver by returning 0. Subsequent invocations of \texttt{cb\textunderscore propagate()} will repeat this process, ensuring that all newly deduced literals are propagated until no further deductions are possible.

	\mysubsection{\texttt{int cb\textunderscore add\textunderscore reason\textunderscore clause\textunderscore lit(int propagated\textunderscore lit)}}
	\label{sec:cb_add_reason_clause_lit}
	
	This method requests a reason clause from the \theory-solver, providing an explanation for the previously propagated literal \texttt{propagated\textunderscore lit}. The explanation provided to the SAT solver is constructed using~\Cref{alg:lemmadeduction}, applied to the \theory-lemma that fixed the phase of the corresponding constraint.
	Similar to the process described in \texttt{cb\textunderscore propagate}, the reason clause is constructed and propagated one literal at a time. The end of the clause is signaled by returning 0.

	\mysubsection{\texttt{bool cb\textunderscore has\textunderscore external\textunderscore clause()} / \texttt{int cb\textunderscore add\textunderscore external\textunderscore clause\textunderscore lit()}}
	\label{sec:cb_external_clause}
	
	These two methods are used to handle conflict clauses. \texttt{cb\textunderscore has\textunderscore external\textunderscore clause()} checks whether a conflict clause was deduced by the \theory-solver during propagation, as discussed in~\Cref{sec:pdcl}.
	If a conflict clause is available, \texttt{cb\textunderscore add\textunderscore external\textunderscore clause\textunderscore lit()} propagates one literal from the conflict clause at a time. The end of the conflict clause is signaled by returning 0.

	\section{Proof Production in Alethe Format}
	\label{sec:proofs}

\picid{} produces proofs in the Alethe format~\cite{BaFlFoSc,ScFlBaFo21} for various reasons: \begin{inparaenum}[(i)]
		\item It is a standard proof format for SMT solving, used by existing tools~\cite{BaBaBrKrLaMaMoMoNiNoOzPrReShTiZo22}. As such, it naturally lends itself to expressing proofs that follow the DPLL(\theory) scheme, integrating proof steps of a SAT solver with those of a \theory-solver;
		\item It is readily supported by an independent proof checker,
		enhancing reliability and scalability~\cite{AnLaBa23}; and
		\item Marabou proofs can be expressed naturally  in this format.
	\end{inparaenum}
	
	The module responsible for proof generation is \emph{\writer},
        which receives information from both the SAT solver and the \theory-solver, and writes corresponding proof steps.
	The \writer{} is implemented within the \theory-solver, and it implements the extension of IPASIR-UP for proof production.
	It receives two types of notifications from the SAT solver and then uses the notified information and  data structures of the \theory-solver to write proofs steps in Alethe.
	The first, $add \_ derived \_ clause$, informs \writer{} that a Boolean resolution proof step is learned by the SAT solver. These steps are given in the LRAT proof format, and can be directly translated into Boolean resolution steps in Alethe.
	The second, $add \_ original \_ clause$, informs \writer{} whenever conflict and reason clauses are learned by the \theory-solver. \writer{} converts the proofs of the given clause, and all \theory-lemmas required for proving that clause, into Alethe format. In particular, proofs of \theory-lemmas are written lazily~\cite{KaBaTiReHa16}, only after \picid{}'s analysis identifies them as a part of the minimized proof.

	\subsection{Expressing Marabou Proofs in Alethe}
	We discuss here the details of converting Marabou proofs into the Alethe format.
	This includes formalizing four main components: \begin{inparaenum}[(i)]
		\item DNN verification problem in the SMT standard format~\cite{BaFoTi16}, readable by Alethe;
		\item Proofs of \unsat{} for each subspace of the search tree;
		\item Proofs of bound propagation based on the \relu{} function, namely \theory-lemmas; and
		\item Proof of soundness of case splitting, based on phases of the activation functions.
	\end{inparaenum}
	
	Since the fourth component is abstracted into Boolean
        predicates and case splitting is handled by the SAT solver,
        the correctness of case splitting in \picid{} proofs is
        established within the LRAT part of the proof. Consequently, we focus on translating the first three components.

	\mysubsection{Formalization in SMT.} 
	\picid{} proofs are checked against a formalization of the DNN verification problem in the SMTLIB standard format~\cite{BaFoTi16}. Given a query $\langle x, A, l, u, R \rangle$, we encode $\langle x, A, l, u \rangle$ as a set of assertions in quantifier-free linear real arithmetic (\textsc{QF\textunderscore LRA}). Equations of the constraint matrix $A$ are represented by assertions of the form:
		$$\small{\bm{\mathtt{(\assert \; (=\; (+ \;(*\; c\textunderscore j1\; x\textunderscore j1) \cdots \;(*\; c\textunderscore jk\; x\textunderscore jk)\; )\; 0.0)) }}}$$
	and bound are represented by inequalities of the form:
		$$\small{\bm{\mathtt{(\assert \; (<= \; x\textunderscore j \; u(x_j)) ) \;\; ; \;\; (\assert \; (>= \; x\textunderscore j \; l( x_ j) ) ) }}}$$
	For the constants $u(x_j),l(x_j)$. Following a standard formulation used in prior work~\cite{KaBaDiJuKo21}, \relu{} 
	constraints of the form $f = \relu(b)$ are encoded as:
	$$\small{\bm{\mathtt{(\assert \; (ite \; (>= \; b \; 0.0) \; (=\; b \; f \;) \; (<= \; f \; 0.0))) }}}$$

	\mysubsection{Proofs by the Farkas Lemma.} 
	Recall that \unsat{} for each subspace of the search tree is proven by the Farkas lemma as in~\Cref{thm:Farkas}. Alethe, on the other hand, supports another variant of the Farkas lemma, captured in the linear arithmetic proof rule $\texttt{la\_generic}$, formalized as follows:
	$$\small{\bm{
	\mathtt{(\step \: i \: (\cl \:  
	\varphi_1 \cdots \varphi_n)\: :\userule \:\: la\_generic \: :\args(a_1 \cdots a_n))}}}
	$$
	
	where $\varphi_i$ are the linear arithmetic literals of the derived clause, and \texttt{args} are the Farkas coefficients used to derive \unsat{} of the negated clause. Notably, every literal in the clause is coupled with a coefficient. In Marabou, however, $w$ describes only the coefficients of rows of $A$ and not of $l,u$.
	
	To translate a proof vector $w$ into the Alethe format, we first compute the linear combination $\sum_{i=1}^{n} c_i \cdot x_i$. We then construct the corresponding clause by negating: \begin{inparaenum}[(i)]
		\item The rows $A_j$ for which the associated coefficient $w_j \neq 0$;
		\item Upper-bound constraints $u_k-x_k \geq 0$ with positive coefficients $c_k$; and
		\item Lower-bound constraints $x_k-l_k \geq 0$ with negation of the negative coefficients $c_k$.
	\end{inparaenum}
	
	To ensure numerical stability, this procedure is implemented using an arbitrary-precision arithmetic library~\cite{GMP}, in contrast to other components of the \theory-solver. Below, we formalize this reduction.
	
	\mysubsection{Reduction of Farkas Vectors.}
	In Alethe, the \texttt{la$\_$generic} rule is checked by defining a sequence of atoms $\varphi_1, \ldots, \varphi_n$ together with corresponding coefficients. Then, the weighted linear combination is reduced to the contradiction $c > 0$ or $c \geq 0$ for some negative constant $c$~\cite{BaFlFoSc}.
	
	To translate Marabou proof vectors as defined in~\Cref{thm:Farkas} into the Alethe format, we instantiate atoms using rows of the constraint matrix $A$ and bound expressions of the form $x_i - l_i \geq 0$ and $u_i - x_i \geq 0$ for lower and upper bounds, respectively. The remaining task is to compute the appropriate coefficients for each of these expressions.
	Formally, we prove:
	\begin{theorem}[Reduction of Marabou Proof Vectors to Alethe]\label{thm:reduction}
		Consider the query $ A\cdot x = \bar{0} $ and $ l \leq x \leq u $, for $ A \in M_{m \times n} (\mathbb{R})$  and $ l,x,u \in \rn{n} $. Let $\mathcal{L}, \mathcal{U}$ be a list of $n$ expressions, whose $j^{th}$ entry is $x_j - l_j, u_j-x_j$ respectively.
		If $ \exists w \in \rn{m} $ such that for 
		$w^\intercal \cdot A \cdot x \coloneq \underset{j=1}{\overset{n}{\sum}}c_j\cdot x_j$, we have that  $ \underset{c_j > 0}{\sum}c_j\cdot u_j +  \underset{c_j < 0}{\sum}c_j\cdot l_j < 0$, then we can construct a proof vector $(w_1,w_2,w_3)\in\rn{m+2n}$ such that  $w_1^\intercal \cdot A \cdot x + w_2^\intercal \cdot \mathcal{L} + w_3^\intercal \cdot \mathcal{U}$ is reduced to a negative constant.
	\end{theorem}
	
	\emph{Proof:}
	Let $w_1\coloneq w$, and consider the linear combination $w^\intercal \cdot A \cdot x \coloneq \underset{j=1}{\overset{n}{\sum}}c_j\cdot x_j$. For $w_2$, its $j^{th}$ entry is $-c_j$ if and only if $c_j < 0$ and zero otherwise. For $w_3$, its $j^{th}$ entry is $c_j$ if and only if $c_j > 0$ and zero otherwise. Observe that:
        \begin{multline*} 
        	w_1^\intercal \cdot A \cdot x + w_2^\intercal \cdot l + w_3^\intercal \cdot u =  \\
        	\underset{j=1}{\overset{n}{\sum}}c_j\cdot x_j + \underset{c_j < 0}{\sum}(-c_j)\cdot (x_j-l_j) + \underset{c_j > 0}{\sum}c_j\cdot (u_j - x_j) =  \\
        	\underset{c_j < 0}{\sum}c_j\cdot x_j  + \underset{c_j < 0}{\sum}c_j\cdot (l_j-x_j) +  \underset{c_j > 0}{\sum}c_j\cdot x_j + \underset{c_j > 0}{\sum}c_j\cdot (u_j - x_j) =  \\
        	\underset{c_j > 0}{\sum}c_j\cdot u_j +  \underset{c_j < 0}{\sum}c_j\cdot l_j
        	\end{multline*}
         We assumed the last term is a negative constant. $\square$
	
		To improve proof compactness, when writing the proof step we use $(w_1,w_2,w_3)$ after omitting all zero coefficients and their corresponding constraints. By definition, these are unnecessary for the proof.
	\begin{example}
		Consider $w = \begin{bmatrix} 1 & -1 & -1 & -1 & 1 \end{bmatrix}^\intercal$ illustrated in~\Cref{fig:proof}, which generates the equation $i_1 - 3i_2 - i_3 +a_1 -a_2 + o = 0$ and deduce \unsat.
		The corresponding Alethe proof step is as follows: Its clause consist of negation literals for all equations of $A$ whose corresponding entry in $w$ is non zero (i.e., all equations), the upper bounds $i_1 \leq 1$, $a_1 \leq 0$, $o \leq -2$ and the lower bounds $0 \leq i_2$, $0\leq i_3$, $0\leq a_2$.
		The corresponding arguments for the proof rule are $(w_1, w_2, w_3)$ for $w_1 =w$, $w_2 = \begin{bmatrix} 3 & 1 & 1  \end{bmatrix}^\intercal$ and $w_3 = \begin{bmatrix} 1 & 1 & 1  \end{bmatrix}^\intercal$. 
	\end{example}
	
	\mysubsection{Proofs of \theory-Lemmas.}
	Recall that \theory-lemmas describe the derivation of a \emph{derived bound} from a previously established \emph{causing bound}, together with a proof rule that justifies this derivation. Consequently, proving a \theory-lemma consists of two parts: first, a proof of the causing bound, and second, Alethe proof steps that encode the corresponding Marabou proof rule.
	The first part is proven in Marabou using the Farkas lemma, and its translation to Alethe follows the reduction described above. The second part relies primarily on Boolean reasoning over \relu{} phases and the bounds they induce. An example for such derivation appears in~\Cref{fig:aletheProof}.
	
	\mysubsection{Handling Holes.}
	Before producing a candidate proof vector $w$, Marabou performs a sanity check based on the Farkas lemma to validate its correctness. If the proof vector passes this check, it is written as part of the proof; otherwise, Marabou deems the constructed proof vector invalid, and the corresponding case split is marked as \emph{delegated} to external, numerically stable solvers. Those, in turn, could either correct the proof, or reveal errors made in the verification process.
	Prior work indicates that such delegation is rare~\cite{IsBaZhKa22, IsReWuBaKa26a}.
	In the Alethe format, this situation is represented using the \texttt{hole} rule,  indicating that the derivation of a given clause is left unproven. Such proof steps are of the following form:
		$$\small{\bm{
		\mathtt{(\step \: j \: (\cl \:  
			\varphi_1 \cdots \varphi_n)\: :\userule \:\:hole)}}}
	$$
	
	\mysubsection{Example.} For completeness, an example for a \picid{} proof in the Alethe format, formalizing the proof of~\Cref{ex:proof}, appears in~\Cref{fig:aletheProof}.
	The proof begins with \emph{assumption} commands representing the constraints of the verification query $Q=\langle x, A, l, u, R \rangle$. Lines 2–6 encode the equations of the matrix $A$, lines 9–19 specify all initial necessary bound constraints $l$ and $u$, and lines 22–23 encode the \relu{} constraints for neurons $v_1$ and $v_2$. The proof uses the \texttt{:named} tag to assign identifiers to terms, improving readability. In particular, we denote equations by $e_j$ and \relu{} constraints by $R_j$ for some index $j$.
	The proof then introduces auxiliary steps, in lines 26-33, that derive equalities such as $u(a_j) = 0$ from $u(b_j) = 0$, and vice versa, which are both valid in the active phase of the corresponding \relu{}.
	Then, lines 36–38 present a sequence of steps establishing a bound-tightening lemma. It begins with an application of the Farkas lemma (\texttt{la\textunderscore generic}), showing that the negation of the causing bound leads to \unsat{}. This is followed by resolution steps deriving both the causing bound and the derived bound.
	Next, in lines 41-51 the proof derives \unsat{} for each leaf of the proof tree. As before, each derivation consists of an application of the Farkas lemma followed by resolution. Finally, the proof concludes by deriving the empty clause, combining the unsatisfiability of both leaves. Note that \texttt{lemma3} is derived, though it is not used as a premise of any other step. In practice, this bound-tightening lemma would be removed from the proof by~\Cref{alg:lemmadeduction}, and its steps will not appear in the Alethe file.

	\begin{figure*}[!t]
		\begin{lstlisting}
; Constraint matrix
(assume e1 (!(= (+ i1 (- i2) (- b1)) 0.0):named e1))
(assume e2 (!(= (+ (* 2.0 i2) i3 (- b2)) 0.0):named e2))
(assume e3 (!(= (+ f2 (- f1) (- o)) 0.0):named e3))
(assume e4 (!(= (+ f1 (- b1) (- a1)) 0.0):named e4))
(assume e5 (!(= (+ f2 (- b2) (- a2)) 0.0):named e5))

; Bound constraints
(assume u_i1 (!(<= i1 1):named u_i1))
(assume l_i1 (!(>= i1 0):named l_i1))
(assume u_i2 (!(<= i2 1):named u_i2))
(assume l_i2 (!(>= i2 0):named l_i2))
(assume u_i3 (!(<= i3 1):named u_i3))
(assume l_i3 (!(>= i3 0):named l_i3))
(assume l_f1 (!(>= f1 0):named l_f1))
(assume l_f2 (!(>= f2 0):named l_f2))
(assume l_a1 (!(>= a1 0):named l_a1))
(assume l_a2 (!(>= a2 0):named l_a2))
(assume u_o (!(<= o -2):named u_o))

; Relu constraints
(assume r1 (!(ite (>= b1 0.0) (= b1 f1) (<= f1 0.0 )):named r1))
(assume r2 (!(ite (>= b2 0.0) (= b2 f2) (<= f2 0.0 )):named r2))

; Auxiliary steps
(step aux1 (cl (>= b1 0.0)(<= f1 0.0)):rule ite1 :premises(r1))
(step aux2 (cl (not (>= b1 0.0))(= b1 f1)):rule ite2 :premises(r1))
(step aux3 (cl (not (= b1 f1))(not e4)(<= a1 0.0)):rule la_generic :args(1 1 1))
(step aux4 (cl (not (>= b1 0.0))(<= a1 0.0)):rule resolution :premises(aux2 aux3 e4))

(step aux5 (cl (not (>= b2 0.0))(= b2 f2)):rule ite2 :premises(r2))
(step aux6 (cl (not (= b2 f2))(not e5)(<= a2 0.0)):rule la_generic :args(1 1 1))
(step aux7 (cl (not (>= b2 0.0))(<= a2 0.0)):rule resolution :premises(aux5 aux6 e5))

; Lemma derivation
(step lemma1 (cl (not e2)(not l_i2)(not l_i3)(>= b2 0.0)):rule la_generic :args(-1 -2 1 1))
(step lemma2 (cl (>= b2 0.0)):rule resolution :premises(lemma1 e2 l_i2 l_i3))
(step lemma3 (cl (<= a2 0.0)):rule resolution :premises(lemma2 aux7))

; Right leaf
(step right1 (cl (not e1)(not e2)(not e3)(not e4)(not e5)(not l_i2)(not l_i3)(not l_a2)
                 (not u_i1)(not u_o)(not (<= a1 0.0))):rule la_generic 
                 :args(1 -1 -1 -1 1 3 1 1 1 1 1))
(step right2 (cl (not (<= a1 0.0))):rule resolution
                 :premises(e1 e2 e3 e4 e5 l_i2 l_i3 l_a2 u_i1 u_o right1))
(step right3 (cl (not (>= b1 0.0))):rule resolution :premises(right2 aux4))

; Left leaf
(step left1 (cl (not e3)(not l_f2)(not u_o)(not (<= f1 0.0))):rule la_generic :args(-1 1 1 1))
(step left2 (cl (not (<= f1 0.0))):rule resolution :premises(left1 e3 l_f2 u_o))
(step left3 (cl (>= b1 0.0)):rule resolution :premises(aux1 left2))

; Finalization
(step final (cl):rule resolution :premises(right3 left3))
		\end{lstlisting}
		\caption{An example of a \picid{} proof in Alethe.}
		\label{fig:aletheProof}
	\end{figure*}
	
	\section{Parallelization}
	\label{sec:snc}
	
	Another key feature of \picid{} is the integration of a
  (SNC)~\cite{WuOzZeJuIrGoFoKaPaBa20} parallelization scheme,
	enabling concurrent solving of queries within the proof producing CDCL(\theory)
	framework.
	The core idea is to split complex queries into $M$ simpler  subqueries, based on common decision heuristics, and handle each of them independently.
	The subqueries are placed in a shared work queue, and a pool of $N$ worker
	threads (typically $N \leq M$) repeatedly retrieves subqueries from this queue
	and solves them independently. This process continues until the queue is
	exhausted.
	
	The global result is determined by aggregating the outcomes of
	all subqueries. If any worker identifies a satisfying assignment for its
	assigned subquery, the original query is immediately declared \sat.
	Conversely, the original query is \unsat{} only after all
	subqueries teminated with \unsat.
	
	To support CDCL(\theory) solving and Alethe proof production within this
	setting, we instantiate a dedicated SAT solver instance per subquery.
	Each solver is initialized with the assumptions corresponding to the specific
	subquery it is responsible for. In parallel, each worker also creates a
	dedicated \writer{} instance, which records the proof steps associated with that
	subquery. In \unsat{} queries, upon the termination of all workers, a single worker
	constructs the final concluding proof steps for the original query. These final
	steps combine the concluding steps of all subquery proofs, into
	a complete proof of \unsat{} for the original verification query.

	\section{Evaluation}
	\label{sec:Evaluation}
	We evaluated \picid{} along two axes: first, its scalability as a reliable verifier, measured by both verification time and \unsat{}-proof checking time; and second, the correctness of the proofs, and their ability to detect flaws in the verification process.
	
	As a baseline, we considered the standard proof-producing variant of Marabou extended with proof generation in Alethe.\footnote{Another possible baseline is the proof-producing version of Marabou used in~\cite{DeIsKoStPaKa25}; however, it is highly unoptimized and thus far less competitive.} To assess the overhead of producing Alethe proofs, we also evaluated Marabou’s native proof-producing version.
	Additionally, we evaluated two configurations of \picid{}: with and without parallelization. This setup allows us to separately assess the impact of CDCL-based reasoning and parallel solving. All generated \unsat{} proofs were checked by \textsc{Carcara} using a single thread.
	
	Finally, we evaluated our tool compared to \textsc{NeuralSAT}~\cite{DuNgDw25}, a state-of-the-art proof-producing DNN verifier. NeuralSAT employs incomplete methods that significantly improve performance though are inherently more complex to prove, resulting in compact and less detailed proofs as we discuss in~\Cref{sec:conclusion}. 
	Due to technical limitations, we were able to successfully
        generate NeuralSAT proofs for only a single benchmark, on which we report. These proofs, written in an ad-hoc format, were validated using NeuralSAT's dedicated checker~\cite{DuNg25}, which employs floating-point arithemtic.
	
	Note that \picid{} could also be compared to proof-producing SMT solvers such as \textsc{cvc5}~\cite{BaBaBrKrLaMaMoMoNiNoOzPrReShTiZo22}. However, prior work has shown that such solvers do not scale beyond small DNNs with only dozens of neurons~\cite{KaBaDiJuKo21}. Our preliminary experiments confirm this observation (e.g., \textsc{cvc5} did not solve any \textsc{ACASXu} query within two hours), and so we omit these results.
	
	All experiments were conducted on CPU machines running Linux Debian~12.
	We allocated a memory limit of 32GB for verification and 64GB
	for proof checking. We imposed a verification timeout of 90 minutes for \textsc{ACASXu} and \textsc{MNIST\_FC}, as they express queries over larger networks, and 
	30 minutes for \textsc{CERSYVE} and \textsc{SafeNLP}. Parallelization experiments were conducted using
	16 cores.

	We evaluate our approach on four benchmark suites used in VNN-COMP~\cite{BrMuBaJoLi23,KaLaBaBrDuFlJoKoMaNgWu25}: \textsc{ACASXu}~\cite{JuKoOw19}, comprising DNNs with 300 neurons; \textsc{CERSYVE}~\cite{YaHuWeLiLi25}, with networks ranging from 64 to 198 neurons; \textsc{SafeNLP}~\cite{CaArDaIsDiKiRiKo23,CaDiKoArDaIsKaRiLe25}, consisting of DNNs with 128 neurons; and \textsc{MNIST\_FC}~\cite{BaLiJo21}, which includes DNNs with 512--1024 neurons. These benchmarks are widely used in the DNN verification literature and were selected to cover a diverse range of network size and application domains.
	
	For \textsc{CERSYVE}, we evaluate all queries appearing in the
	benchmark suite over two commits of VNN-COMP benchmarks~\cite{VNNCOMP25}.
	For \textsc{MNIST\_FC}, the original benchmark suite comprises of queries over three fully connected network architectures with
	2$\times$256, 4$\times$256, and 6$\times$256 hidden layers. Here, we consider only the queries corresponding to the two smaller architectures.
	Moreover, each \textsc{MNIST\_FC} property is expressed as a disjunction of nine
	constraints. As disjunction is not yet natively supported in
        \picid, we preprocess each query by breaking it down  into nine subqueries, one per disjunct.

	\begin{table*}[t]
		\centering
		\small
		\setlength{\tabcolsep}{4pt}
		\renewcommand{\arraystretch}{0.95}
		
		\caption{Evaluation results per benchmark. We report the number of solved
			queries, average verification and checking times, and proof validity
			(\textsc{valid}/\textsc{holey}/\textsc{invalid}/\textsc{out of resources}/\textsc{not generated }). Best performing results across comparable measurements are marked in \textbf{bold}. Avg. Shared represent averaging results for queries solved by all methods. }
		\scalebox{1}{
			\vspace{-0.2cm}
			\begin{tabular}{l cc cc cc c c c c c c}
				\toprule
				\multirow{2}{*}{\textbf{Method}}
				& \multicolumn{2}{c}{\textbf{Solved (\#)}}
				& \multicolumn{2}{c}{\textbf{Avg. Verification (s)}}
				& \multicolumn{2}{c}{\makecell{\textbf{Avg. Shared} \\ \textbf{Verification (s)}}}
				& \multirow{2}{*}{\textbf{Avg. Check (s)}}
				& \multirow{2}{*}{\makecell{\textbf{Avg. Shared} \\ \textbf{Check (s)}}}
				& \multirow{2}{*}{\makecell{\textbf{Validity} \\ (V/H/I/O/N)}} \\
				\cmidrule(lr){2-3}\cmidrule(lr){4-5}\cmidrule(lr){6-7}
				& \textbf{\sat} & \textbf{\unsat} & \textbf{\sat} & \textbf{\unsat} & \textbf{\sat} & \textbf{\unsat}& \\
				\midrule
				
				\multicolumn{7}{l}{\textsc{ACAS Xu} (180 queries)}\\
				\midrule
				Marabou & 40 & 87  & 327.71 & 591.71 & 327.71 & 365.89 & N/A & N/A& N/A \\
				Marabou + Alethe& 40 & 82  & 467.06 & 522.05 & 467.06 & 522.06   & 395.45 & 395.45& 82/0/0/0/0 \\
				\picid{} & 42 & 95  & 475.65 & 631.34 & 241.38 & 226.46& 503.03& \textbf{113.68}& 94/1/0/0/0 \\
				\picid{} with SNC & \textbf{45} & 110 & 203.59 & 618.17 & \textbf{59.43} & 128.34 & 1014.09 & 287.5 & 98/1/0/11/0 \\
				NeuralSAT & \textbf{45} & \textbf{135}  & 123.75 & 76.05 & 101.57 & \textbf{48.70} & 2815.53 & 730.58 & 58/0/0/22/55 \\
				\midrule
				
				\multicolumn{7}{l}{\textsc{CERSYVE} (20 queries)}\\
				\midrule
				Marabou & \textbf{10}
				 & 8  & 15.04 & 226.37 & 15.04 & 104.39 & N/A & N/A& N/A \\
				Marabou + Alethe& \textbf{10} & 7 & 30.32 & 173.83
				 & 30.32 & 173.83 & 103.25 & 103.25& 5/2/0/0/0 \\
				\picid & \textbf{10} & 8  & 10.56 & 128.33 & 10.56 & 64 & 45.37 & \textbf{27.02} & 5/1/2/0/0 \\
				\picid{} with SNC & \textbf{10} & \textbf{10} & 17.96 & 157.67 & \textbf{3.86} & \textbf{29.68}& 71.04 & 27.28 & 6/2/2/0/0 \\
				\midrule
				
				\multicolumn{7}{l}{\textsc{SafeNLP} (1080 queries)}\\
				\midrule
				Marabou & 577  & 209  & 2.89 & 117.6 & 2.89 & 60 & N/A & N/A& N/A \\
				Marabou + Alethe & 578 & 202 & 5.89 & 93.83 & 2.95 & 93.83  & 61.82 & 61.82 & 156/0/0/0/46 \\
				\picid & 620 & 246 & 19 & 165.11 & 1.36 & 18.7  & 67.20 & \textbf{10.48}& 200/0/0/0/46 \\
				\picid{} with SNC & \textbf{632} & \textbf{263}  &8.06 & 123.59 & \textbf{1.3} & \textbf{9.18} &243.67 &14.66 & 216/0/0/1/46 \\
				\midrule
				
				\multicolumn{7}{l}{\textsc{MNIST\_FC} (540 queries)}\\
				\midrule
				Marabou & 35 & 165  & 752.99 & 322.55 & 487.97 & 290.39 & N/A & N/A& N/A \\
				Marabou + Alethe & 36 & 165  & 957.06 & 367.99  &559.33 & 342.58 & 9.26 & 8.35 & 157/0/0/0/8 \\
				\picid & 25 & 209  & 844.81 & 367.32 & 844.81 & 75.19 &  18.3 & \textbf{2.90} &201/0/0/0/8 \\
				\picid{} with SNC & \textbf{54} & \textbf{215}  & 852.58 & 520.29 & \textbf{149.55} & \textbf{50.39}  &263.19 & 7.88 & 207/0/0/0/8 \\
				\bottomrule
			\end{tabular}
		}
		\label{tab:results}
	\end{table*}

	Table~\ref{tab:results} summarizes the results for all
	benchmarks and configurations. Additional graphs detailing our results appears in Appendix~\ref{app:res}. Below, we analyze our results with respect to the two evaluation questions. Note that some queries were deduced \unsat{} during preprocessing, across all methods, and thus their proofs are not generated.
	
	\mysubsection{Performance Speed.} Overall, \picid{} significantly
	outperforms Marabou in most cases, both in terms of the
	number of solved queries and in average verification and proof-checking times. In particular, \picid{} solved queries $2.3\times$-$5\times$ faster for \unsat{}
	queries, and demonstrates mixed results of $0.66\times$-$2.9\times$ improvement factor for \sat{} queries
	on average.
	When considering SNC, performance is further improved, leading to an increased number of solved queries and reduced average verification
	times by $4\times$-$10.2\times$ for \unsat{} queries and by $2.3\times$-$7.8\times$ for \sat{} queries on average.
	However,
	we observe that proof checking times are higher for \picid{} with SNC. We hypothesize this is due to the fact that worker threads contribute independent proof fragments, which are subsequently combined into a single proof. Therefore, enabling workers to share clauses and their proofs may lead to a decrease in both proof size and checking speed. We leave this to future work.
	
	As expected, NeuralSAT outperforms \picid{} both in the number
        of solved queries in their average running time. However,
        \picid{} successfully produced
        proofs for a greater number of \unsat{} instances, and these
        proofs were checked in a significantly shorter time. This comparison emphasizes the trade-off between scalability of DNN verifiers, and their ability to produce detailed proofs.

        Comparing the Alethe-producing baseline to the native
        proof-producing Marabou, we observe that producing Alethe proofs incurs an overheads that increases verification time by a factor of $1.18\times$-$1.6\times$ for \unsat{} queries and by $1.02\times$-$2\times$ for \sat{} queries on average. This overhead is expected, due to the involved file writing and the computation of the vector reduction as in~\Cref{thm:reduction}. However, this overhead is significantly diminished by the performance improvements of \picid{}.
	
	\mysubsection{Proof Correctness.}
	Overall, on three out of four benchmarks, all proofs whose checking terminated within the allotted time and memory constraints were found correct by \textsc{Carcara}, besides a single unstable query of \textsc{ACASXu}, as previously known~\cite{IsBaZhKa22, IsReWuBaKa26a}. 	The \textsc{CERSYVE} benchmark, however, exhibits cases where \picid{} and Marabou are numerically unstable. This is reflected in three scenarios:
	\begin{inparaenum}[(i)]
		\item In several cases, both tools detect this instability, resulting in the generation of proof holes. These proofs are reported as \emph{holey}, meaning all other proof steps were correctly checked;
		\item  Two queries that are solved as \unsat{} only by \picid{}, yield invalid proofs, detected by \textsc{Carcara}; and
		\item In one instance, Marabou reports \sat{} while \picid{} incorrectly reports \unsat~\cite{KaLaBaBrDuFlJoKoMaNgWu25}, accompanied with a holey proof. The holes of the proof correspond to the potential subspaces that were incorrectly solved, possibly due to numerical instability, thus narrowing the possible states that lead to error.
	\end{inparaenum}
	
	Overall, \picid{} was able to generate valid proofs for more
        than $99\%$ of the cases. In all other cases, the
        proof-checking mechanism successfully detected the flaws, either by failing to construct correct proofs, or by generating proof holes that pinpoint the potentially erroneous steps. 
	
	\section{Conclusion}
	\label{sec:conclusion}
	We present \picid, a proof-producing DNN verifier that implements a parallel CDCL(\theory) architecture in which conflict clauses are derived from \unsat{} proofs. \picid{} is the first DNN verifier to generate proofs in a standard SMT format. Notably, \picid{} does not scale as state-of-the-art verifiers, since it does not employ symbolic, incomplete verification methods, which are intrinsically complex to prove. Nonetheless, \picid{} pushes forward both scalability and standardization of proof-producing DNN verification.
	
	\mysubsection{Related Work.}
	Several DNN verifiers employ CDCL(\theory)-like verification schemes~\cite{LiYaZhHu24,DuNgDw25,Eh17,ZhBrHaZh24}, leveraging the analogy between \relu{} activation phases and Boolean assignments. Notably, these verifiers either rely on naive clause-learning or first detect \unsat{} cases and subsequently reinvoke solving procedures to minimize clause size~\cite{ChDr91,ZhBrHaZh24}. In contrast, our approach directly exploits information from \unsat{} proofs to derive non-trivial conflict clauses.
	
	One verifier of \emph{Binarized} Neural Networks (BNNs)~\cite{YaTaSoMyMe25}, uses a SAT-based approach to verify, produce and reliably check proofs. This work is tailored to BNNs while our work focuses on verifying a broader class of DNNs. 
	
	Most closely related to our work is the NeuralSAT DNN verifier~\cite{DuNgDw25}, which employs a CDCL(\theory) architecture and is able to produce \unsat{} proofs~\cite{DuNg25}. However, NeuralSAT proofs are not directly comparable to ours: they are checked through repeated invocations of an MILP solver, and require additional elaboration before they can be translated into a standard proof format. Furthermore, the MILP solver involved in the checking process uses floating-point arithmetic, and thus it prone to the same reliability issues as NeuralSAT.

	An additional approach for producing proofs in neural network verification is presented in~\cite{ZhMeAnLuPaSoPi26}. It characterizes a different settings, in which the (Zero-Knowledge) proofs are used to convince that verification results are correct without disclosing the DNN weights. 
	
	Other approaches to reliable DNN verification
        are presented in~\cite{SiSaMeSi25,ToSyMu25,AlJaVo26}. The work in~\cite{SiSaMeSi25} seeks to prove
        the correctness of abstract-interpretation-based DNN
        verifiers, formalized in a dedicated domain-specific language. 
        In contrast, \picid{} produces SMT-like proofs for individual
        verification queries, while making fewer assumptions regarding
        the internal workings of the verifier. Consequently, the two
        approaches are not directly comparable.
        The work in~\cite{ToSyMu25} presents a formally-verified algorithm that verifies the Lipschitz-continuity-based robustness of DNNs. Although it provides a reliable verification procedure, it targets a single class of properties and relies on an incomplete method. In contrast, \picid{} is complete and supports verification of any \textsc{QF\textunderscore LRA} property. 
        The Rocq-NN-Roll prover, recently presented in~\cite{AlJaVo26}, is a DNN verifier implemented and formally verified in Rocq. It implements reliable verification procedures as well, though it scales up to DNNs containing only a few neurons.
	
	Finally, a variety of parallelization techniques have been explored in SMT solving and DNN verification~\cite{WuOzZeJuIrGoFoKaPaBa20,WaZhXuLiJaHsKo21}. To the best of our knowledge, our approach is the first technique for DNN verification that inherently produces proofs.
	
	\mysubsection{Future Work.} We plan to extend \picid{} along several directions. 
	First, we plan to extend proof production to incomplete verification methods, enabling support for a broader range of optimizations and non-linear activation functions.
	Second, we aim to support additional piecewise-linear activation functions, and produce proofs for Marabou's preprocessing.
	Third, we aim to further reduce \picid's proof size by allowing workers to share learned clauses together with their associated proofs, reducing redundant derivations of similar lemmas.
	We believe that combining these directions will yield a faster and more versatile proof-producing verifier, capable of scaling to larger DNNs than is currently possible.
	
	\mysubsection{Acknowledgments.} The work of Isac, Refaeli and Katz was supported by the European Union (ERC, VeriDeL, 101112713). Views and opinions expressed are however those of the author(s) only and do not necessarily reflect those of the European Union or the European Research Council Executive Agency. Neither the European Union nor the granting authority can be held responsible for them.
	The work of Wu is supported by a gift from VMware University Research Fund.
	The work of Barrett was supported in part by the National Science Foundation (grant number  2211505), and the Stanford Center for AI Safety.
	
	The authors wish to thank Armin Biere, Katalin Fazekas, Mathias Fleury, and Mathias Priener for the guidance on using IPASIR-UP, and to Bruno Andreotti and Haniel Barbosa for their guidance on generating Alethe proofs and using Carcara.  

	\bibliographystyle{abbrv}
	\bibliography{CDCL}

\renewcommand\thesection{\Alph{section}}
\renewcommand\thesubsection{\thesection.\arabic{subsection}}
\setcounter{section}{0}
\newpage
\noindent\begin{minipage}{\textwidth}
{\noindent\huge{Appendix}}
\raggedbottom
\section{Additional Analysis of Results}
\label{app:res}
	\centering
	\setlength{\tabcolsep}{4pt}
	\renewcommand{\arraystretch}{0.85}
	\scalebox{0.85}{
	\begin{tabular}{c cc}
		
		\raisebox{0.8\height}{\rotatebox{90}{\textsc{ACAS Xu}}} &
		\includegraphics[width=0.45\linewidth]{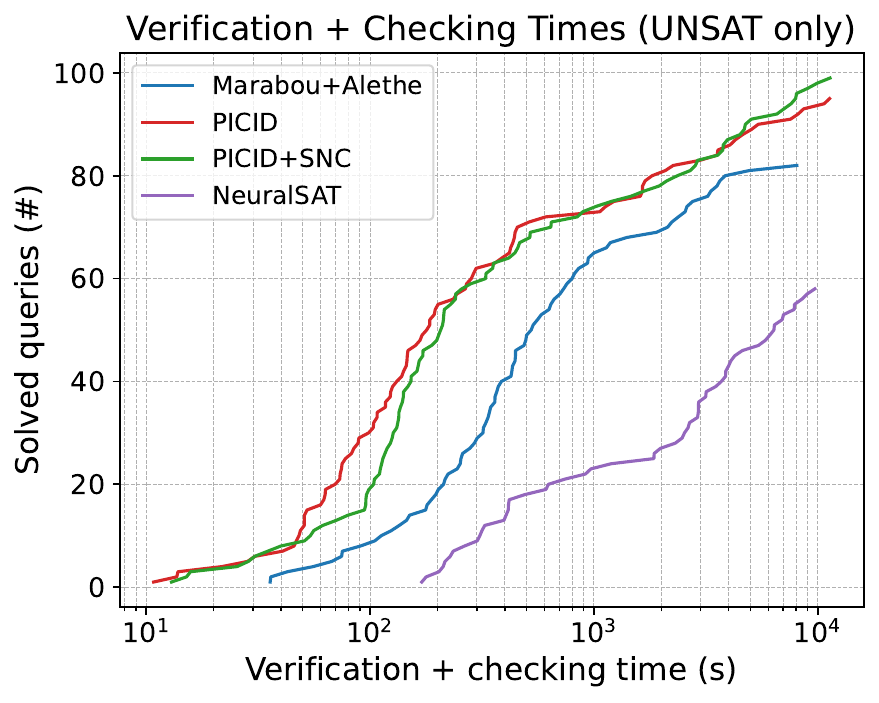} &
		\includegraphics[width=0.45\linewidth]{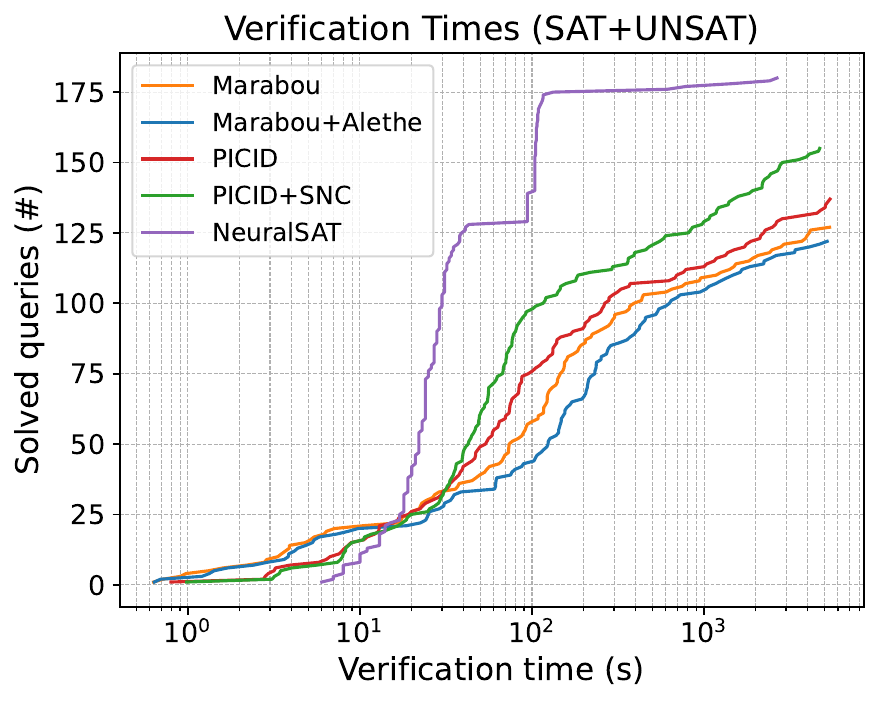} \\
		
		\raisebox{0.8\height}{\rotatebox{90}{\textsc{CERSYVE}}} &
		\includegraphics[width=0.45\linewidth]{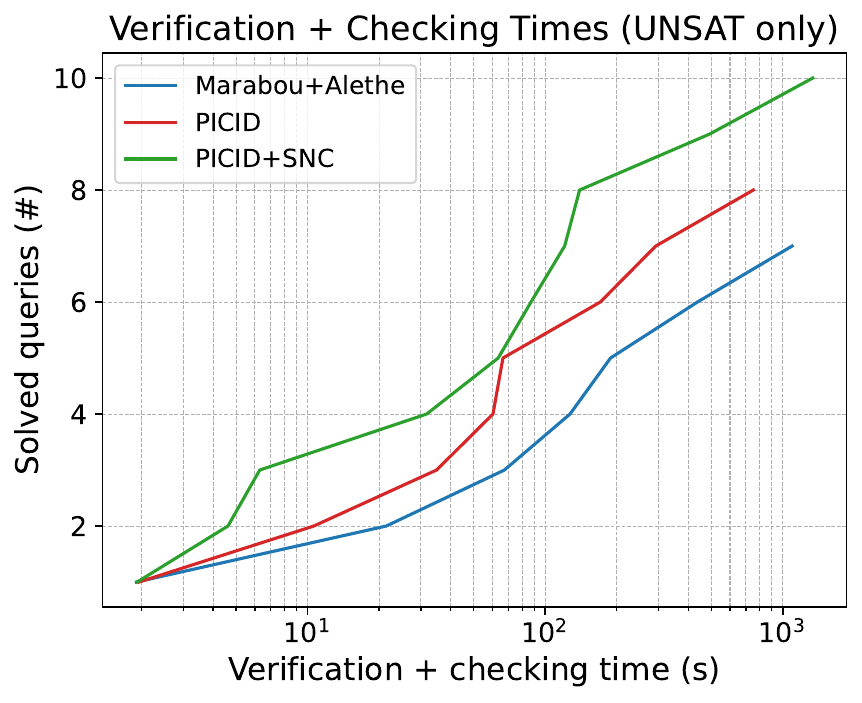} &
		\includegraphics[width=0.45\linewidth]{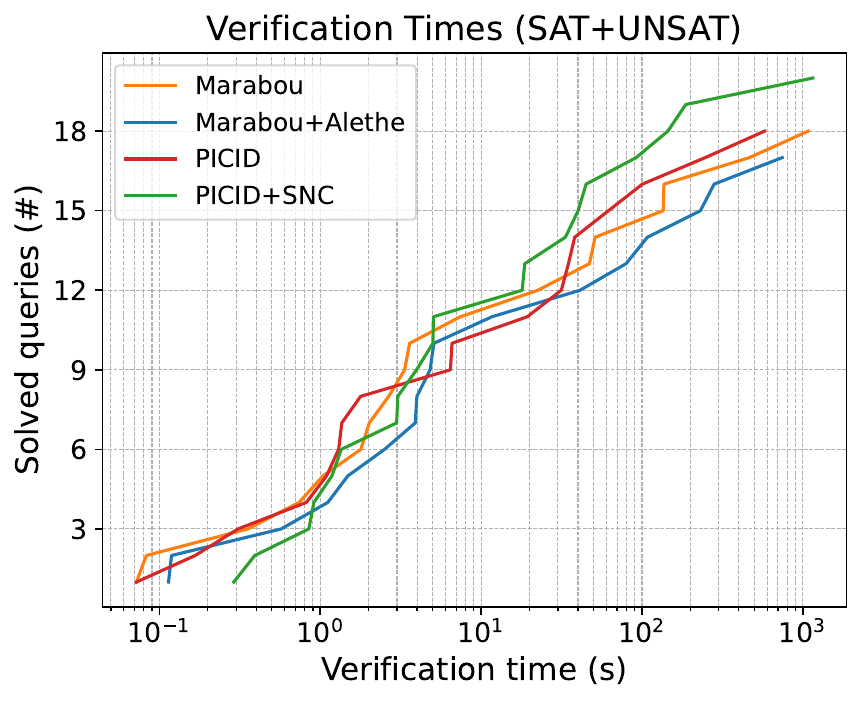} \\
		
		\raisebox{0.8\height}{\rotatebox{90}{\textsc{SafeNLP}}} &
		\includegraphics[width=0.45\linewidth]{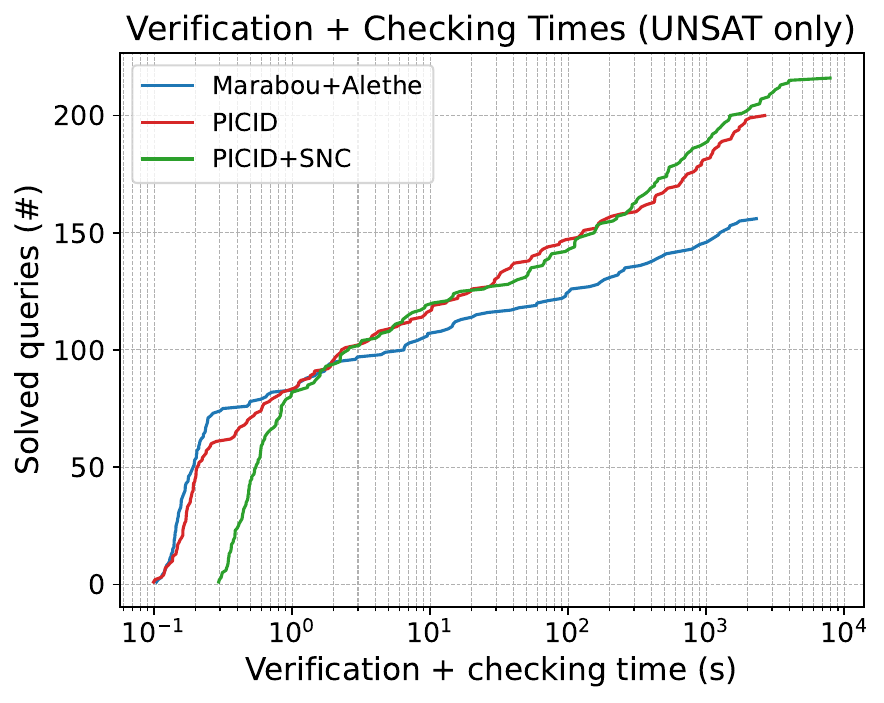} &
		\includegraphics[width=0.45\linewidth]{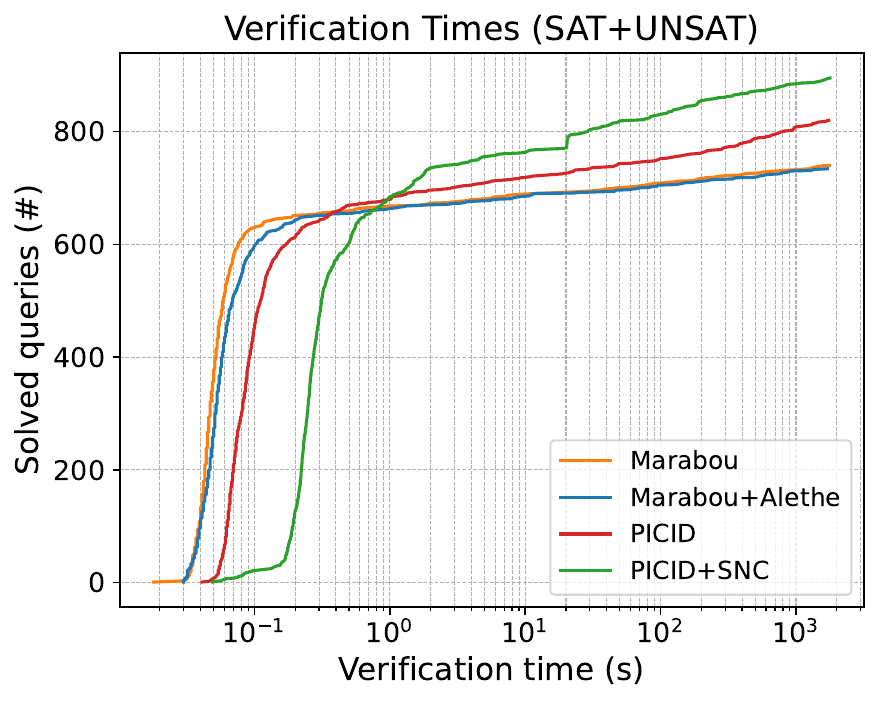} \\
		
		\raisebox{0.8\height}{\rotatebox{90}{\textsc{MNIST\_FC}}} &
		\includegraphics[width=0.45\linewidth]{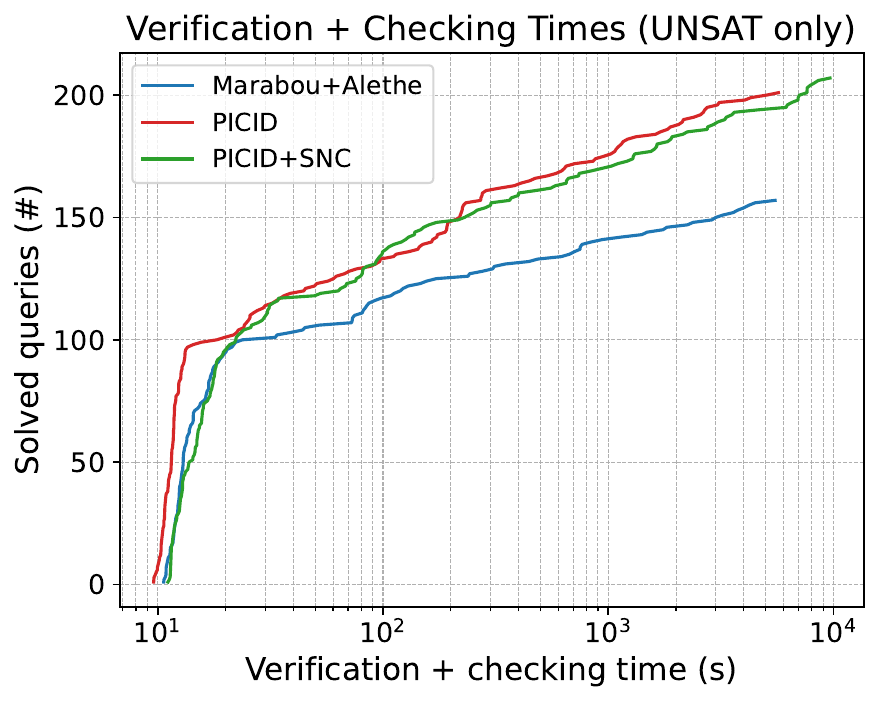} &
		\includegraphics[width=0.45\linewidth]{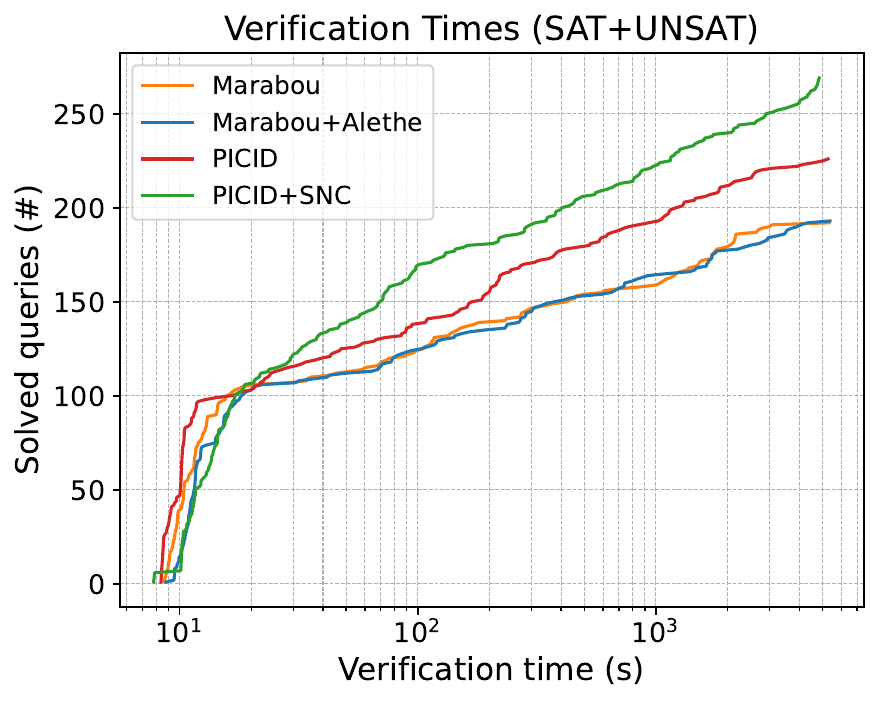} \\

	\end{tabular}

	}

	\captionof{figure}{Cactus plots for all benchmarks.
		\emph{Left column:} cumulative number of solved and proved \unsat{} queries as a function of total time (verification + checking). \emph{Right column:} cumulative number of solved queries (\sat{} and \unsat{}) as a function of verification time. Time measurements are discretized to a one-second resolution.}
	\label{fig:appendix-cactus}
	
\end{minipage}

\end{document}